\documentclass[12pt]{article}
 
\usepackage[letterpaper,margin=1in]{geometry}
\usepackage{setspace}
\pdfoutput=1
\usepackage[T1]{fontenc}
\usepackage[utf8]{inputenc}
\usepackage{lmodern}
\usepackage{lineno}
\usepackage{microtype}
\usepackage{amsmath, amssymb, amsthm}
\usepackage{bm}
\usepackage{booktabs}
\usepackage{array}
\usepackage{graphicx}
\usepackage{xcolor}
\usepackage{caption}
\usepackage[round, authoryear]{natbib}
\usepackage{hyperref}
\hypersetup{colorlinks=true, linkcolor=black, citecolor=black, urlcolor=blue}
\usepackage{tikz}
\usetikzlibrary{positioning, arrows.meta, decorations.pathreplacing, calc}
 
\newcolumntype{L}[1]{>{\raggedright\arraybackslash}p{#1}}
 
\newtheorem{proposition}{Proposition}

\theoremstyle{definition}
\newtheorem{definition}{Definition}
\newtheorem{example}{Example}
\theoremstyle{remark}
\newtheorem{remark}{Remark}
 
\newcommand{\sv}{\mathbf{s}}
\newcommand{\pv}{\mathbf{p}}
\newcommand{\Lam}{\boldsymbol{\Lambda}}      
\newcommand{\wv}{\mathbf{w}}
\newcommand{\R}{\mathbb{R}}
\newcommand{\zmode}{z}              
\newcommand{\Phimech}{\Phi}         
\newcommand{\pinertia}{\chi}        
\newcommand{\clip}{\mathrm{clip}}

\begin{document}

\title{\bfseries Indirect Geoeconomic Influence:\\
A Switching Dynamical Systems Framework for Mechanism Design}

\author{
  Nikolos Gurney\thanks{Corresponding author: \texttt{gurney@ict.usc.edu}}
  \and Boxi Fu
  \and Soham Hans
  \and Volkan Ustun
}

\date{
  Institute for Creative Technologies, University of Southern California\\[4pt]
  \today
}

\maketitle

\renewcommand{\thefootnote}{\fnsymbol{footnote}}
\footnotetext[2]{FUNDING: Research was sponsored by the Army Research Office and was accomplished under Cooperative Agreement Number W911NF-25-2-0040. The views and conclusions contained in this document are those of the authors and should not be interpreted as representing the official policies, either expressed or implied, of the Army Research Office or the U.S. Government. The U.S. Government is authorized to reproduce and distribute reprints for Government purposes notwithstanding any copyright notation herein.}
\renewcommand{\thefootnote}{\arabic{footnote}}

\begin{abstract}
\noindent
We develop a formal framework for analyzing \emph{indirect} geoeconomic
influence. The influencing state (sender) does not attempt to change a 
target nation's policy directly. Instead, the sender restructures the 
target's internal political economy so that its own citizens, firms, and 
institutions generate the compliance pressure. The framework rests on a 
switching dynamical system (SDS) in which a target's political economy 
evolves under mode-dependent rules. We analyze two modes: a permissive 
mode, in which a mechanism transmits pressure toward the sender's
preferred policy, and a contested mode, entered naturally once 
the target detects and attributes the mechanism. Crucially, the
sender's mechanism design shapes the transition into the contested
mode. Our approach inverts the usual regime-switching problem in which 
analysts estimate a latent transition kernel from data. The indirect influence 
approach is to have the designer engineer a kernel to steep regime occupancy 
over a planning horizon. A structured \emph{switch vector} decomposes any 
mechanism along discrete design dimensions, and a combinatorial optimizer 
searches this space for high-performing archetypes scored on compliance, 
time-to-threshold, and a durability ratio. We characterize mode-conditional 
equilibria and derive comparative statics on credibility and legibility, 
showing that the legibility penalty is scaled by the salience of the 
government channel and interacts with the mechanism's cost incidence. We 
illustrate the framework with two stylized mechanisms and report a two-part 
proof of concept: a simulation of 972 archetypes, and a fully crossed study 
of the pipeline's optional generative stage in which language models write 
and blindly classify candidate instruments for archetypes. 
\end{abstract}

\medskip
\noindent\textbf{Key words:} geoeconomic influence; mechanism design;
switching dynamical systems; hybrid systems; economic statecraft;
decision analysis for policy design

\bigskip

\section{Introduction}
\label{sec:intro}

Nations increasingly pursue geopolitical objectives through economic
means. Trade policy \citep{scholvin2018geo}, investment screening
\citep{Danzman2024eu}, technical standards \citep{zuniga2024geopolitics},
and financial infrastructure \citep{FarrellNewman2019} have become
instruments of statecraft deployed not to produce economic outcomes in
isolation but to engineer the incentive environments of target states.
The resulting literature on geoeconomic influence has produced important
taxonomic and structural insights
\citep{Baldwin2020,Drezner1999,FarrellNewman2019}. However, it has
largely treated the target nation as a unitary actor that responds to
external pressure as a single strategic entity. This abstraction
obscures the internal transmission structure through which influence
actually operates. That structure provides the pathways by which
external mechanisms route pressure through a target's civil society,
private sector, and government factions to produce a policy outcome.

This paper develops a formal framework for analyzing indirect geoeconomic 
influence (IGI), in which a state (the sender) seeks a policy outcome in 
a target nation not by mandating it directly, but by engineering the
target's internal political economy so that its own citizens, firms, and
institutions generate the compliance pressure. The framework decomposes
influence mechanisms along discrete structural dimensions --- choices
about how pressure is timed, routed, and sustained within the
target --- and a combinatorial optimizer searches this space to identify
high-scoring mechanism configurations. It is a \emph{generative} tool, 
meaning it is designed to discover and evaluate novel mechanisms before 
deployment. 

The framework's central modeling assumption is that a target's political
economy under influence does not evolve under a single set of rules. For
example, before a mechanism is detected and attributed, the target's
government channel transmits pressure from external sources like its
private and civil-society channels do: the sender's investment loads
momentum that moves policy toward the sender's preferred outcome. Once
the target's government identifies the mechanism, the same channel
\emph{reverses}: the government becomes a net source of counter-pressure,
and its resistance capacity rises as it mobilizes. These are
qualitatively distinct dynamical regimes with an endogenous boundary.
It is a function of the mechanism's design --- its legibility, external 
support, and cost incidence, among other design choices --- together with 
the starting conditions. Modeling this boundary, and handling the sender 
design choices over it, is the analytical core of our paper.

This orientation inverts the standard use of switching dynamical
systems. In the tradition initiated by \citet{Hamilton1989}, one
observes a time series generated by a system that operates under
distinct regimes and estimates, after the fact, both the latent regime
sequence and the transition kernel governing it. The IGI framework runs
this machinery backward. The transition kernel is not estimated from
data but \emph{designed}: the sender's switch vector shapes when the
target enters a resistant regime and the policy problem becomes one of 
steering regime occupancy. Our approach facilitates, for example, 
maximizing time spent accumulating compliance in the permissive mode 
before the contested mode engages. To our knowledge this reorientation 
is novel in the geoeconomic setting.

The framework is developed for the interstate case but its logic extends
to any setting in which an actor induces behavioral change in a complex
system through indirect incentive design. Inducing private-sector firms
to internalize national-security considerations in sourcing decisions
without direct mandates is structurally identical: an actor engineers an
incentive environment, pressure routes through internal subsystems, and
a policy outcome emerges endogenously. Supply-chain resilience,
conditional development finance, and multilateral coordination all fall
within scope. 

The framework is positioned as a decision-analytic tool for designing
policy interventions where randomized evidence is unavailable and
outcomes depend on an adversary's endogenous response. We argue this problem
fits squarely within the concerns of structured decision analysis and
adversarial risk analysis \citep{riosinsua2009adversarial}. It is intended 
as a prescriptive tool to enumerate, compare, and stress-test candidate 
policy responses before deployment.

\section{Related Work}
\label{sec:lit}

The study of geoeconomic influence draws on several distinct literatures
\citep{mohr2025geoeconomics}. Although no single one provides the formal
dynamic framework this paper develops, each contributes theoretical
grounding. We situate our contribution within four threads: economic
statecraft, the formal modeling of sanctions and coercion,
regime-switching and hybrid dynamical systems, and mechanism design.

\subsection{Economic Statecraft}
\label{sec:lit.econ}
\citet{Baldwin2020} provides the foundational vocabulary, organized by 
instrument type (sanctions, aid, trade policy), and evaluative framework 
for the field. \citet{Aggarwal2021} traces how that vocabulary has been 
stretched by contemporary practice, as statecraft has moved from discrete 
instruments toward the reshaping of trade and technology regimes 
themselves. Baldwin argues that the effectiveness of economic instruments 
should be assessed comparatively and against a counterfactual, as the 
standard of full compliance is almost never realized. We adopt his basic 
framing of influence as an attempt by a sender to alter a target's 
behavior through economic means, and his assertion that instrument 
choice be evaluated on cost-effectiveness grounds. Baldwin's framework 
does not, however, provide a formal dynamic model of how instruments 
produce compliance, nor is it organized by the structural dimensions of 
how influence propagates through a target's internal political economy. 
We do both and propose a complementary decomposition at a finer level of
structural resolution.

\citet{FarrellNewman2019} offer a theoretically developed contribution
in which states and their relationships are modeled as a network. They
argue that asymmetric network structures concentrate coercive power in
hub nodes, and that states with jurisdiction over those hubs can
weaponize interdependence through the panopticon effect (hub control
generates surveillance advantages) and the chokepoint effect (hub
control enables denial of access). Their framework is structural and
external: it explains why hub states possess leverage but treats the
target's response as a function of network position, not internal
political dynamics. They do not offer a theory of the transmission 
structure that routes external pressure through a target's subsystems 
toward a policy outcome. We pick up where they leave off, proposing a 
formal model of that internal transmission structure, and, critically, 
of how it changes once the target reacts.

\citet{BlackwillHarris2016} offer a practitioner-oriented survey of
geoeconomic tools and argue that the United States has underutilized them
relative to rising powers. Their taxonomy of instrument categories is
descriptively useful for decomposing the structural dimensions of
influence propagation, but does not provide a formal theoretical
account.

\subsection{Formal Modeling of Sanctions and Coercion}
\label{sec:lit.formal}
\citet{Drezner1999} provides a rigorous game-theoretic treatment of
economic sanctions. His conflict-expectations model explains why
sanctions are imposed frequently despite a poor track record, by showing
that adversarial relationships generate incentives to sanction that are
independent of the probability of success. As Drezner states, the model
rests on the assumption that governments act as rational unitary actors.
This abstraction forecloses analysis of the internal transmission
structure through which external pressure reaches a policy outcome.
Drezner's model has no civil society, no private sector, and no
government factions. The central move of our framework is to relax this 
assumption, decomposing the target into interacting subsystems whose
differential response to external mechanisms determines the outcome.

\citet{kaempfer1988theory} represent the most direct attempt to open the 
target's internal political economy to formal analysis. Working in the 
public-choice tradition, they model interest groups within the target as 
political demanders and suppliers of policy, and ask how external 
sanctions affect the equilibrium among them. They offer an insight that 
anticipates the core of our model: the political effects of sanctions
depend on how economic pressure distributes across domestic 
constituencies. However, their apparatus is static. It characterizes a 
political equilibrium and performs comparative statics on that set point. 
It has no time dimension, no accumulation of durable capital, no 
mechanism momentum, and no notion that the target's response rules change
once it reacts. The IGI framework extends their intuition into a dynamic,
mode-switching setting in which the routing and persistence of pressure 
plus the timing of the target's reaction determine the outcome. 

More recently, \citet{liou2023pressures} demonstrate empirically that 
sanctions alter the bargaining environment between target governments and 
domestic civil-society campaigns. They document the civil-society 
transmission channel that our pressure vector formalizes, but they do 
not model the dynamics of that transmission. \citet{Drezner2015} 
documents a parallel private-sector channel in which targeted financial 
measures operate less through direct state compliance than through the 
risk calculations of banks and intermediaries.

Analyzing variation in the effectiveness of U.S. pressure on Japanese 
trade policy, \citet{schoppa1993two} argues that direct strategies 
systematically underperform relative to synergistic strategies that 
deliberately activate and amplify domestic reform coalitions within the 
target. His concept of reverberation in which a sender's demands 
strengthen the hand of domestic reformers is a qualitative description of 
the indirect-influence logic we formalize. However, it differs from the 
IGI approach in that its typology is qualitatively inferred from case 
comparison, it is specific to bilateral trade negotiations, and it is 
taxonomic classifying observed strategies without providing a method for 
generating novel ones.

\subsection{Regime-Switching and Hybrid Dynamical Systems}
\label{sec:lit.sds}
\citet{Hamilton1989} establishes the methodological foundation on which
this paper builds. His insight is that a single set of dynamic equations
is insufficient to describe a macroeconomic system that behaves under
qualitatively different rules at different times. For example, gross
domestic product dynamics switch between an expansion regime and a
recession regime. The key contribution for our purposes is not the
specific Markov-switching apparatus but the underlying philosophy: some
systems are better characterized as operating under distinct behavioral
regimes than as following a single continuous dynamic. Political economies
under geoeconomic pressure exhibit this structure, switching between 
permissive and contested regimes as a mechanism is detected. 

\citet{Kim1994} develops the filtering machinery for this class, 
inferring from data the regime-conditional state distribution by 
collapsing it into a mixture weighted by regime probabilities. Our 
expected-trajectory formulation (Section~\ref{sec:modeswitch}) performs 
the analogous collapse in the generative direction: mode occupancy is 
propagated forward as a probability and the dynamics blended accordingly.

\citet{GhahramaniHinton2000} demonstrate that switching state-space models
are productive analytical tools well outside economics, applying 
variational inference to switching dynamics. When the mode boundary is 
driven by an accumulating state, the object is a \emph{hybrid} dynamical 
system in the sense of \citet{van2007introduction}: continuous dynamics 
within modes, transitions governed by that state. The proof-of-concept in 
this paper tracks the expected trajectory under an exposure-driven 
detection hazard. We note the deterministic-guard and full-stochastic 
variants as extensions.

\citet{Hamilton1989} and \citet{GhahramaniHinton2000} work in the 
inferential direction to describe change. They fit switching models to 
observed data to recover latent regimes and estimate transition 
probabilities after the fact. This paper works in the opposite direction. 
A vector of design choices is specified ex ante by a policy designer who 
enumerates, compares, and simulates mechanism configurations before 
deployment, and those choices shape the transition structure itself. The 
result is a framework that is generative and prescriptive. This 
reorientation from using an SDS to explain observed dynamics to using it 
to \emph{design} desired ones is, to our knowledge, novel in the
geoeconomic context.

\subsection{Mechanism Design}
\label{sec:lit.mech}
The concept of a designer engineering the incentive environment such that 
the target's own actors generate the desired outcome is not itself novel. 
The formal mechanism-design literature has worked on engineering policies 
that have the desired outcome as an equilibrium for over fifty years 
\citep{Hurwicz1973}. Although the IGI framework does not operate in the 
formal mechanism-design tradition (we do not specify complete games
with typed agents, strategy spaces, and implementation theorems n the 
sense of \citet{Maskin1999}), the intuition is the same. Country~$A$ does 
not set Country~$B$'s policy; it restructures $B$'s incentive environment 
so that $B$'s own citizens, firms, and institutions generate the 
compliance pressure. 

Classical mechanism design asks what outcomes are implementable in 
principle. The algorithmic tradition \citep{NisanRoughgarden2007} asks 
what a designer can actually compute. The IGI framework is more aligned 
with the latter, trading the generality of a full implementation result 
for an enumerable design space a policy analyst can search exhaustively.

\section{The IGI Framework}
\label{sec:framework}

We model geoeconomic influence as an indirect engineering problem.
Country~$A$ (the sender) seeks a policy outcome in Country~$B$ (the
target) but does not directly set $B$'s policy. Instead, $A$ deploys a
\emph{mechanism} that loads momentum into $B$'s system through targeted
behavioral channels. This momentum routes through $B$'s internal
subsystems (government, civil society, and the private sector) and
resolves against $B$'s institutional resistance capacity. The policy
outcome is endogenous to $B$'s own dynamics, not imposed from outside,
and $B$'s dynamics are not fixed. Once $B$ detects and attributes the
mechanism, its government channel switches from compliant to defiant.

\subsection{Overview of the Pipeline}
\label{sec:pipeline}

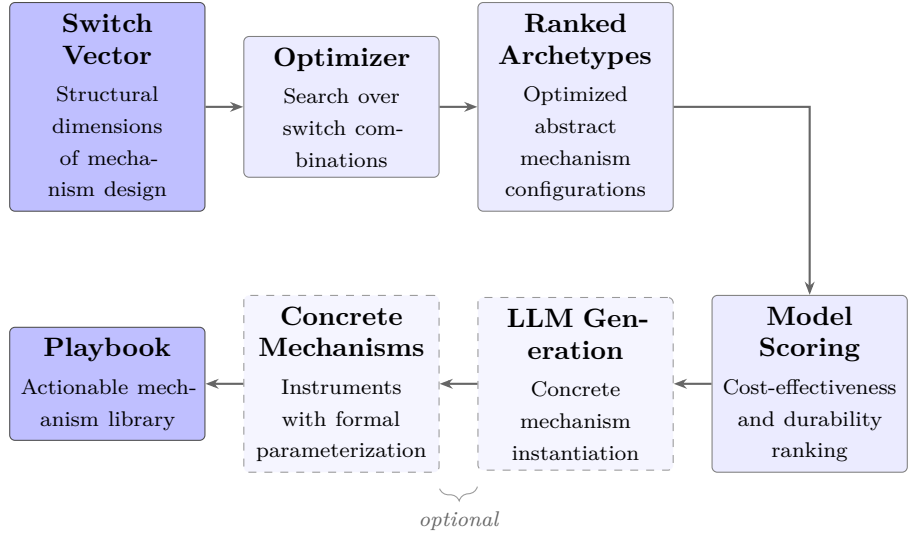
\begin{figure}[ht]
  \centering
  \begin{tikzpicture}[
      node distance = 0.45cm and 0.5cm,
      every node/.style = {font=\small},
      primary/.style = {rectangle, rounded corners=2pt, draw=black!70,
        fill=blue!25, text width=2.3cm, align=center, minimum height=1.1cm, inner sep=4pt},
      stage/.style = {rectangle, rounded corners=2pt, draw=black!50,
        fill=blue!8, text width=2.3cm, align=center, minimum height=1.1cm, inner sep=4pt},
      optional/.style = {rectangle, rounded corners=2pt, draw=black!40,
        fill=blue!4, text width=2.3cm, align=center, minimum height=1.1cm, inner sep=4pt, dashed},
      arr/.style = {-{Stealth[length=5pt]}, thick, draw=black!60},
  ]
  \node[primary] (sv) {\textbf{Switch Vector}\\[2pt]
     \scriptsize Structural dimensions of mechanism design};
  \node[stage, right=of sv] (opt) {\textbf{Optimizer}\\[2pt]
     \scriptsize Search over switch combinations};
  \node[stage, right=of opt] (ra) {\textbf{Ranked Archetypes}\\[2pt]
     \scriptsize Optimized abstract mechanism configurations};
  \node[stage, below right=1.1cm and 0.5cm of ra] (ms) {\textbf{Model Scoring}\\[2pt]
     \scriptsize Cost-effectiveness and durability ranking};
  \node[optional, left=of ms] (llm) {\textbf{LLM Generation}\\[2pt]
     \scriptsize Concrete mechanism instantiation};
  \node[optional, left=of llm] (cm) {\textbf{Concrete Mechanisms}\\[2pt]
     \scriptsize Instruments with formal parameterization};
  \node[primary, left=of cm] (pb) {\textbf{Playbook}\\[2pt]
     \scriptsize Actionable mechanism library};
  \draw[arr] (sv)  -- (opt);
  \draw[arr] (opt) -- (ra);
  \draw[arr] (ra.east) -| (ms.north);
  \draw[arr] (ms)  -- (llm);
  \draw[arr] (llm) -- (cm);
  \draw[arr] (cm)  -- (pb);
  \draw[decorate, decoration={brace, amplitude=5pt}, draw=black!40, thin]
    ([yshift=-6pt]llm.south west) -- ([yshift=-6pt]cm.south east)
    node[midway, below=5pt, font=\scriptsize\itshape, text=black!60] {optional};
  \end{tikzpicture}
  \caption{The IGI framework pipeline. Dark boxes are primary inputs and
    outputs; light boxes are intermediate processing stages; dashed boxes
    are optional. The pipeline transforms a structured switch vector
    through combinatorial optimization into scored, actionable mechanism
    archetypes. An optional language-model generation stage translates
    archetypes into concrete candidate policy instruments;
    Section~\ref{sec:llm-stage} tests whether an archetype can be
    recovered from the instrument generated for it.}
  \label{fig:pipeline}
\end{figure}

The pipeline (Figure~\ref{fig:pipeline}) operates as follows. Any
influence mechanism is decomposed into structural dimensions, represented
as a \emph{switch vector} $\sv$. A combinatorial optimizer searches the
Cartesian product of switch values for archetypes, or optimized abstract
mechanisms, that maximize predicted compliance under the model of
Section~\ref{sec:model}. High-scoring archetypes are scored on
cost-effectiveness metrics (Section~\ref{sec:assessment}) and,
optionally, translated by a language-model generation stage into
candidate policy-instrument descriptions.

\subsection{The Switch Vector}
\label{sec:switch}

\begin{definition}[Switch Vector]
A \emph{switch vector} $\sv = (s_1, \ldots, s_d)$ is a $d$-dimensional
vector in which each component $s_i$ takes a value from a finite
candidate set $\mathcal{S}_i$. The Cartesian product $\mathcal{S} =
\prod_{i=1}^{d} \mathcal{S}_i$ defines the \emph{mechanism space}: the
set of all abstract mechanism configurations under the assumed
decomposition.
\end{definition}

Table~\ref{tab:switches} presents eleven notional dimensions, drawn from 
a review of geoeconomic instruments, that illustrate an architectural 
decomposition of geoeconomic mechanisms. The list is not meant to be 
definitive. Not all switches are active in every mechanism: 
\emph{contingent} switches may take null values and drop from the model; 
\emph{structural} switches (e.g., locus of pressure) are always active.

Two dimensions risk being conflated. \emph{Locus of pressure} identifies 
which of $B$'s internal subsystems (government, private sector, civil 
society, mixed) bears the primary pressure. It is the routing choice and 
what the model's pressure vector (Section~\ref{sec:pressure}) is defined 
over. \emph{Target population} identifies which stratum of $B$'s society 
the mechanism addresses (elite, mass, institutional, networked). It 
conditions how readily pressure originating at that stratum reaches the 
routed subsystem, but it does not itself determine routing. A mechanism 
may, for example, address a networked professional population while 
routing pressure through the private sector.

\begin{table}[ht]
  \centering
  \caption{Notional Mechanism Switch Vector}
  \label{tab:switches}
  \renewcommand{\arraystretch}{1.2}
  \begin{tabular}{L{3.8cm} p{9.3cm}}
    \toprule
    \textbf{Switch} & \textbf{Candidate Values} \\
    \midrule
    Influence timing       & Latent / Active / Mixed \\
    Target population      & Elite / Mass / Institutional / Networked \\
    Propagation pathway    & Direct / Cascade / Systemic \\
    Behavioral channel     & Divestment / Alignment / Mobility / Consumption \\
    Credibility mechanism  & Legal / Institutional / Market / Demonstrated \\
    Threat structure       & Automatic / Discretionary / Graduated \\
    Momentum persistence   & Low / Medium / High \\
    Feedback polarity      & Stabilizing / Destabilizing / Neutral \\
    Locus of pressure      & Government / Private sector / Civil society / Mixed \\
    Cost bearer            & Country $A$ / Country $B$ / Shared \\
    Legibility to $B$      & Transparent / Obscured / Deniable \\
    \bottomrule
  \end{tabular}
  \smallskip\\
  {\small\textit{Note.} Each row defines one candidate dimension of the
    switch vector $\sv$. The optimizer searches over combinations of
    these values to identify high-scoring archetypes. These dimensions
    are notional.}
\end{table}

\subsection{Optimization over Mechanism Space}
\label{sec:optimization_overview}
For a given geopolitical objective, an optimizer searches $\mathcal{S}$
for switch vectors that maximize a compliance objective (defined in
Section~\ref{sec:opt_problem}). The discrete, enumerable structure of
$\mathcal{S}$ means that for small switch spaces exhaustive search is
feasible; for larger spaces, gradient-free methods (genetic algorithms,
Bayesian optimization) trade coverage for speed. The proof-of-concept in
Section~\ref{sec:poc} uses exhaustive enumeration over a reduced space.
High-scoring archetypes serve as abstract blueprints that specify
\emph{how} influence is structured, routed, and sustained, without yet
specifying a concrete policy instrument. An interpretation stage, by
policy experts or via a language-model generation process, then
contextualizes each archetype into candidate instrument descriptions
reflecting the institutional specifics of the target.

\section{Formal Model}
\label{sec:model}

The formal model is a minimal switching dynamical system designed to
distinguish mechanisms along the dimensions of the switch vector. It
rests on three continuous state variables, a discrete mode variable, a
pressure vector representing $B$'s internal political economy, and a set
of switch-dependent functional parameters. The defining feature relative
to the sanctions-modeling literature is that $B$'s response rules are
\emph{mode-dependent}, and the mode transition is endogenous to and
shaped by $A$'s mechanism design.

\subsection{State Space and Modes}
\label{sec:states}
Three continuous state variables characterize the system at each period
$t$:
\begin{itemize}
  \item $y_t \in [0,1]$: Country $B$'s \emph{policy position}, where $0$
        is full resistance to $A$'s preferred outcome and $1$ is full
        compliance.
  \item $k_t \in \R_{\geq 0}$: \emph{structural entrenchment stock} ---
        accumulated leverage that persists independently of $A$'s ongoing
        investment (installed standards, certified exporters, asset
        positions held by $B$'s population).
  \item $m_t \in \R_{\geq 0}$: \emph{mechanism momentum} --- the
        mechanism's current activation level, which may persist after $A$
        reduces investment (high momentum) or decay rapidly (low
        momentum).
\end{itemize}

The system additionally carries a discrete \emph{mode}
$\zmode_t \in \{\mathrm{P}, \mathrm{C}\}$: the permissive mode
($\mathrm{P}$), in which the mechanism operates before the target has
detected and attributed it, and the contested mode ($\mathrm{C}$), in
which the target's government has identified the mechanism and mobilized
against it. The mode governs which rules the continuous states obey
(Section~\ref{sec:dynamics}), and the transition $\mathrm{P}\to\mathrm{C}$
is driven by an accumulating \emph{exposure stock}
$E_t \in \R_{\geq 0}$ (Section~\ref{sec:modeswitch}).

$A$'s mechanism is activated through a non-negative action variable
$a_t \geq 0$ representing investment intensity at time $t$, with total
investment bounded by budget $\bar{C}$: $\sum_{t} a_t \leq \bar{C}$.

For interpretive purposes, $y_t$ admits a categorical overlay: $[0,0.2)$
active resistance; $[0.2,0.5)$ passive non-compliance; $[0.5,0.8)$
partial compliance; $[0.8,1.0]$ full compliance. These bands are labels
on the continuous state, not additional modes; the switching structure of
the model lives in $\zmode_t$, not in the $y$-bands.

\subsection{Internal Pressure Vector}
\label{sec:pressure}
The model distinguishes three internal subsystems whose behavior 
transmits $A$'s mechanism to $B$'s policy outcome:
\begin{itemize}
  \item $p_t^G$: pressure from or within $B$'s \emph{government},
        including internal faction dynamics and resistance capacity.
  \item $p_t^P$: pressure from $B$'s \emph{population} and civil society.
  \item $p_t^F$: pressure from $B$'s \emph{private sector} and firms.
\end{itemize}
These form the vector $\pv_t = (p_t^G, p_t^P, p_t^F)^\top$. The
\emph{routing vector}
$\Lam(\sv,\zmode_t) = (\lambda^G, \lambda^P, \lambda^F)^\top \in \R^3$
distributes scalar mechanism momentum across the three pressure channels.
Components may be negative and the routing depends on the mode: the
government channel that transmits pressure in the permissive mode
reverses sign in the contested mode.\footnote{We treat $\Lam$ as a vector
instead of a matrix because momentum $m_t$ is scalar in the present
specification; a richer model with a vector-valued momentum state would
promote $\Lam$ to a matrix without altering the interpretation.} The 
locus of pressure switch (Table~\ref{tab:switches}) selects how a
mechanism's momentum is distributed across these three channels. 
Note that the fourth value, mixed, distributes pressure across the three 
internal subsystems. The pressure vector has three components; the 
four-valued switch determines their relative loading, not their number. 
In general, this demonstrates how pressure vectors can function as 
discrete or relative instruments.

\subsection{Mode-Dependent Dynamics}
\label{sec:dynamics}
Within a given mode, the three continuous states evolve according to
\begin{align}
  y_{t+1} &= \pinertia\, y_t + \wv(\sv)^\top \pv_t
             - \beta(\zmode_t)\, r(y_t, k_t),
  \label{eq:policy} \\[3pt]
  k_{t+1} &= k_t + \mu(\sv)\, a_t - \kappa\, k_t,
  \label{eq:entrenchment} \\[3pt]
  m_{t+1} &= \rho(\sv)\, m_t + \nu(\sv)\, k_t + \eta(\sv)\, a_t,
  \label{eq:momentum}
\end{align}
with $y_{t+1}$ clamped to $[0,1]$ each period. Here
$\pinertia \in (0,1]$ is the \emph{policy-position inertia}: the fraction
of the current policy position carried into the next period. The baseline
analytic treatment sets $\pinertia = 1$, so that policy position is a
pure accumulator and \eqref{eq:policy} reduces to an increment on $y_t$;
the simulation retains $\pinertia$ as a free parameter
(Section~\ref{sec:poc}).\footnote{We use $\pinertia$ for this parameter 
to avoid collision with the routing components 
$\lambda^G,\lambda^P,\lambda^F$ of Section~\ref{sec:pressure}.}

The weight vector $\wv(\sv) = (w^G, w^P, w^F)^\top$ gives the salience of
each subsystem for the policy outcome. Its dependence on $\sv$ is a
modeling assumption: the salience is taken to follow the \emph{cost
bearer} switch, so that when $A$ bears the mechanism's cost weight
concentrates on the government channel ($w^G$ large), when $B$ bears it
weight shifts toward $B$'s own firms and population ($w^G$ small), and
shared cost weights the channels evenly. This encodes the intuition that
whoever funds a mechanism shapes which of $B$'s subsystems it most
directly loads; like the other switch-to-parameter mappings, it is a
calibrated rather than a derived quantity. In \eqref{eq:policy} the 
resistance term $r(y_t,k_t)$ satisfies $\partial r/\partial y_t > 0$ 
(resistance rises with compliance) and $\partial r/\partial k_t < 0$ 
(entrenchment erodes resistance). The scalar $\beta(\zmode_t) > 0$ is $B$'s 
baseline resistance capacity, and it is mode-dependent: $\beta(\mathrm{C}) 
\geq \beta(\mathrm{P})$, reflecting that a government that has attributed 
the mechanism mobilizes resistance capacity it was not previously expending. 
In \eqref{eq:entrenchment} structural stock accumulates with $A$'s 
investment at rate $\mu(\sv)>0$ and decays at rate $\kappa>0$. In 
\eqref{eq:momentum} momentum persists at rate $\rho(\sv)\in[0,1)$, is fed 
by structural stock at rate $\nu(\sv)$, and receives direct injection from 
$A$'s action at rate $\eta(\sv)$.

The pressure vector is
\begin{equation}
  \pv_t = \Lam(\sv, \zmode_t)\, m_t\, \delta,
  \label{eq:pressure_vec}
\end{equation}
where $\delta \in [0,1]$ is $A$'s credibility as perceived by $B$'s
internal actors. The routing vector is constructed in two steps. A base
routing $\Lam_0(\sv)$ is determined by the locus-of-pressure switch; then
a legibility adjustment is applied to the government component:
\begin{equation}
  \lambda^G(\sv,\zmode_t) =
  \begin{cases}
    \lambda_0^G(\sv) + g\bigl(\ell(\sv)\bigr), & \zmode_t = \mathrm{P},\\[4pt]
    \lambda_0^G(\sv) - \gamma, & \zmode_t = \mathrm{C},
  \end{cases}
  \label{eq:routing}
\end{equation}
where $\ell(\sv)$ is the legibility switch value, ordered
$\mathrm{Deniable} \prec \mathrm{Obscured} \prec \mathrm{Transparent}$;
$g(\cdot)\leq 0$ is a non-positive legibility modifier, non-increasing in
that order (more transparent mechanisms provoke more government
resistance even before full attribution); and $\gamma > 0$ is the
sign-reversing penalty that engages when the target enters the contested
mode. The population and private-sector components $\lambda^P, \lambda^F$
are mode-invariant in the baseline specification.\footnote{Encoding
legibility in the routing structure is the more structurally meaningful
choice because transparency activates targeted government counter-pressure
rather than uniformly damping every channel. Locating it in the
routing avoids double-counting against the exposure dynamics of
Section~\ref{sec:modeswitch}. It also yields the sharper comparative
static of Proposition~\ref{prop:legibility}. A scalar value uniformly
applied would not accomplish the same.}

Equation~\eqref{eq:routing} makes precise the sense in which the model
switches. In the permissive mode the government channel transmits or
mildly resists, if the mechanism is legible. In the contested mode, it
reverses. The magnitude of the swing, and how close a permissive-mode
mechanism already sits to the sign boundary, are both design choices
encoded in $\sv$.

\begin{remark}[Functional forms]
  \label{rem:functional_forms}
  The specific forms of $r(y_t,k_t)$, the switch-to-parameter mappings
  $\mu,\rho,\nu,\eta,\Lam_0,\wv$, the legibility modifier $g(\cdot)$, the
  contested-mode penalty $\gamma$, and the resistance ratio
  $\beta(\mathrm{C})/\beta(\mathrm{P})$ are not fixed by the framework.
  The proof-of-concept (Section~\ref{sec:poc}) adopts the illustrative 
  form $r(y_t,k_t)=y_t/(1+\alpha k_t)$ and notional lookup-table values 
  solely to demonstrate tractability. 
\end{remark}

\subsection{Endogenous Mode Switching}
\label{sec:modeswitch}
The transition from permissive to contested is the analytical core of the
framework and the locus of the sender's design leverage. $B$ accumulates
an exposure stock driven by mechanism activity and scaled by legibility:
\begin{equation}
  E_{t+1} = E_t + \phi\bigl(\ell(\sv)\bigr)\,\bigl(m_t + \omega\, a_t\bigr),
  \label{eq:exposure}
\end{equation}
where $\phi(\ell)\geq 0$ is increasing in legibility (transparent
mechanisms are detected faster; deniable ones accrue exposure slowly),
and $\omega \geq 0$ weights the visibility of direct investment relative
to ambient momentum. Exposure drives a detection \emph{hazard}: the
per-period probability that $B$ attributes the mechanism and transitions
to the contested mode is
\begin{equation}
  h_t \;=\; 1 - \exp\!\bigl(-E_t/\tau\bigr),
  \label{eq:hazard}
\end{equation}
where $\tau>0$ is a detection scale (larger $\tau$ means slower
detection). Writing $\pi_t$ for the probability that the system has
entered the contested mode by period $t$, $\pi_t$ evolves as
$\pi_{t+1} = \pi_t + (1-\pi_t)\,h_t$ and is monotone non-decreasing, since
the transition is irreversible. The reported trajectories are the
resulting \emph{expected} trajectories: the effective routing and
resistance at each period interpolate between the permissive and contested
regimes in proportion to $\pi_t$,
\begin{equation}
  \Phimech^{\mathrm{eff}}_t = (1-\pi_t)\,\Phimech(\sv,\mathrm{P})
    + \pi_t\,\Phimech(\sv,\mathrm{C}),
  \qquad
  \beta^{\mathrm{eff}}_t = (1-\pi_t)\,\beta(\mathrm{P})
    + \pi_t\,\beta(\mathrm{C}).
  \label{eq:blend}
\end{equation}
This expected-trajectory formulation is the deterministic reduction of a
stochastic detection process; it keeps the simulation cheap while letting
dwell time vary smoothly with legibility. We report the period at which
$\pi_t$ first crosses $0.5$ as the mechanism's \emph{dwell time}.

$E_t$ is non-decreasing and, for any mechanism with $\phi(\ell) > 0$ 
under sustained investment, unbounded; the detection hazard $h_t$ 
therefore approaches $1$ and the contested-mode probability $\pi_t \to 
1$; detection is not something a design can avoid, only defer. The design 
question is \emph{when}, not \emph{whether}. Additionally, because 
$\phi$ depends on legibility and the drivers of $E_t$ include both 
momentum and investment intensity, the mechanism designer directly shapes 
the \emph{dwell time} in the permissive mode. This is the design lever 
the framework is built around: aggressive, high-visibility behavior 
accumulates compliance quickly but also raises the contested-mode 
probability faster; a deniable, slow-loading mechanism buys a longer 
permissive window at the cost of slower accumulation.

\begin{remark}[Expected trajectory, deterministic guard, and full hazard]
  \label{rem:hazard}
  The model reports the \emph{expected} trajectory under the detection
  hazard \eqref{eq:hazard}: mode occupancy is tracked by the probability
  $\pi_t$ and the dynamics are blended via \eqref{eq:blend}, which is
  deterministic to compute and yields smooth dependence of dwell time on
  legibility. Two variants bound this choice. Setting the transition to
  fire the first time $E_t$ crosses a threshold recovers a
  \emph{deterministic guard} (a hybrid system with a single switching
  surface \citep{van2007introduction}), while sampling the transition 
  time from the hazard and averaging over realizations gives the full
  \emph{stochastic} form, closer to the Markov-switching lineage of
  \citet{Hamilton1989}. The expected-trajectory reduction used here 
  approximates the mean of the stochastic form and avoids Monte-Carlo 
  simulation over switching times; we note the two variants as extensions.
\end{remark}

\subsection{Within-Mode Equilibrium}
\label{sec:equilibrium}
Because the mode switches at most once and irreversibly, the system
admits a clean two-phase analysis: a within-mode fixed point exists for
each mode, and the trajectory is a permissive-mode transient toward
$(y^*_\mathrm{P}, k^*, m^*)$ that, once the contested mode engages, is 
redirected toward the contested-mode fixed point 
$(y^*_\mathrm{C}, k^*, m^*)$.
Throughout this subsection we take investment to be constant at
$a_t = a^*$, as in the uniform-investment specification of
Section~\ref{sec:opt_problem}. Setting $\Delta y=\Delta k=\Delta m=0$
within a mode, and noting that
\eqref{eq:entrenchment}--\eqref{eq:momentum} are mode-invariant, the
stock and momentum fixed points are pinned by $a^*$:
\begin{align}
  k^* &= \frac{\mu(\sv)\,a^*}{\kappa},
  \label{eq:k_star}\\[4pt]
  m^* &= \frac{\bigl(\nu(\sv)\,\mu(\sv)/\kappa + \eta(\sv)\bigr)\,a^*}{1-\rho(\sv)},
  \label{eq:m_star}
\end{align}
with $\rho(\sv)<1$ required for $m^*$ finite. Define the mode-dependent
effective pressure scalar
\begin{equation}
  \Phimech(\sv,\zmode) \;=\; \wv(\sv)^\top \Lam(\sv,\zmode).
  \label{eq:phi}
\end{equation}
With $\pinertia = 1$, the within-mode equilibrium policy position
satisfies
\begin{equation}
  \Phimech(\sv,\zmode)\, m^*\, \delta \;=\; \beta(\zmode)\, r(y^*_\zmode, k^*).
  \label{eq:equil_condition}
\end{equation}
With $r(y,k)=y/(1+\alpha k)$, \eqref{eq:equil_condition} solves in closed
form within each mode:
\begin{equation}
  y^*_\zmode \;=\;
  \clip_{[0,1]}\!\left(
    \frac{\Phimech(\sv,\zmode)\, m^*\, \delta\, (1+\alpha k^*)}{\beta(\zmode)}
  \right).
  \label{eq:ystar}
\end{equation}
The clip reflects the $y\in[0,1]$ constraint and binds at both ends:
sufficiently high-influence combinations would saturate at full
compliance in the permissive mode, and mechanisms whose government-channel
reversal drives $\Phimech(\sv,\mathrm{C})$ negative are floored at zero in
the contested mode. Where either bound binds, it must be accounted for
when interpreting rankings. For $\pinertia < 1$ the stationarity condition is
$(1-\pinertia)\,y^*_\zmode = \Phimech(\sv,\zmode)\,m^*\,\delta -
\beta(\zmode)\,r(y^*_\zmode,k^*)$, and \eqref{eq:ystar} generalizes to
$y^*_\zmode = \clip_{[0,1]}\!\big(\Phimech(\sv,\zmode)\,m^*\,\delta /
[\,(1-\pinertia) + \beta(\zmode)/(1+\alpha k^*)\,]\big)$, recovering
\eqref{eq:ystar} as $\pinertia \to 1$; the comparative statics below are
unaffected in sign.

Equation~\eqref{eq:equil_condition} supports comparative statics. Because
$\Phimech$ and $\beta$ are mode-dependent, each statement below is a
within-mode claim. The framework's predictions concern how the 
permissive-mode transient and the contested-mode fixed point compare,
and how long the system dwells in each.

\begin{proposition}[Credibility]
  \label{prop:credibility}
  Within either mode, and for $\Phimech(\sv,\zmode) > 0$, equilibrium
  compliance $y^*_\zmode$ is weakly increasing in $A$'s credibility
  $\delta$.
\end{proposition}
\begin{proof}[Proof sketch]
  The left-hand side of \eqref{eq:equil_condition} is linear and
  increasing in $\delta$ when $\Phimech > 0$; with
  $\partial r/\partial y>0$, the fixed-point condition forces
  $y^*_\zmode$ upward to restore balance. Uniqueness follows from
  monotonicity of $r$ in $y$ given the sign restrictions of
  Section~\ref{sec:dynamics}. The claim is weak rather than strict
  because of the clip in \eqref{eq:ystar}. When
  $\Phimech(\sv,\zmode) < 0$, as may obtain in the contested mode, the
  direction reverses: credibility amplifies counter-pressure.
\end{proof}

\begin{proposition}[Legibility and cost incidence]
  \label{prop:legibility}
  In the permissive mode, equilibrium compliance $y^*_\mathrm{P}$ is
  weakly decreasing in legibility (under the ordering of
  Section~\ref{sec:dynamics}), and the size of the legibility penalty
  is proportional to the salience $w^G(\sv)$ of the government channel.
  Consequently, mechanisms whose costs fall on $B$ (small $w^G$) are
  partially insulated from the legibility penalty; in the limiting case
  $w^G=0$, legibility does not affect $y^*_\mathrm{P}$.
\end{proposition}
\begin{proof}[Proof sketch]
  A more legible mechanism lowers $\lambda^G$ via the non-positive
  modifier $g(\cdot)$ in \eqref{eq:routing}. Its effect on the effective
  pressure scalar is $\partial\Phimech/\partial\lambda^G = w^G$ from
  \eqref{eq:phi}, so the reduction in $\Phimech$ is $w^G\,|\Delta g|$ for
  a one-step increase in legibility. Substituting into
  \eqref{eq:equil_condition} and using $\partial r/\partial y>0$ gives
  the monotonicity claim, with magnitude scaled by $w^G$. When $w^G=0$
  the government channel carries no weight and legibility drops out of
  \eqref{eq:phi} entirely.
\end{proof}

\begin{remark}[A falsifiable structural claim]
  Proposition~\ref{prop:legibility} asserts that the cost of being seen 
  depends on who bears the mechanism's cost. Evidence on how detection 
  affects mechanisms of differing cost incidence could undermine this 
  claim. That the $w^G=0$ case makes legibility costless is a sharp, 
  testable prediction of the routing encoding.
\end{remark}

\begin{proposition}[Dwell time and durability]
  \label{prop:durability}
  Fix the permissive-mode effective pressure $\Phimech(\sv,\mathrm{P})$.
  Then cumulative compliance is increasing in permissive-mode dwell time,
  which is in turn decreasing in legibility $\ell(\sv)$ and in investment
  visibility $\omega$. For mechanisms with $\rho(\sv)\to 1$, compliance
  accumulated during the permissive window persists into the contested
  mode, so durability is governed jointly by dwell time and momentum
  persistence.
\end{proposition}
\begin{proof}[Proof sketch]
  From \eqref{eq:exposure}--\eqref{eq:hazard}, the exposure stock rises
  more slowly when $\phi(\ell)$ is small (deniable) and when the drivers
  $m_t+\omega a_t$ are small, so the detection hazard stays low longer;
  thus dwell time falls in legibility and in
  investment visibility. Compliance accrues in the permissive mode at a
  rate set by $\Phimech(\sv,\mathrm{P})$; holding that rate fixed, a
  longer window integrates more compliance before the contested mode
  engages. As
  $\rho\to1$, momentum, and hence pressure, decays slowly after the
  transition, so permissive-mode gains are retained. The conditioning on
  $\Phimech(\sv,\mathrm{P})$ is essential: dwell time is not valuable in
  itself, and a mechanism that loads no pressure gains nothing from a
  long permissive window. A full statement requires the within-mode
  transient, given in closed form for the PoC specification.
\end{proof}

Propositions~\ref{prop:credibility}--\ref{prop:durability} together
express the framework's central tension: aggressive designs accumulate
compliance quickly but raise the detection hazard sooner, while deniable,
slow-loading designs preserve the permissive window. This trade-off is
what generates a non-trivial cost-effectiveness frontier and what makes
mechanism selection genuinely scenario-dependent.

\subsection{Optimization Problem}
\label{sec:opt_problem}
Given a mechanism design $\sv$ and budget $\bar{C}$, $A$ selects an
action path $\{a_t\}$ to maximize cumulative compliance over the horizon:
\begin{equation}
  \max_{\sv,\,\{a_t\}} \; \sum_{t=0}^{T} y_t
  \quad\text{s.t.}\quad \sum_{t=0}^{T} a_t \leq \bar{C},\quad a_t \geq 0,
  \label{eq:optproblem}
\end{equation}
where the trajectory $\{y_t\}$ is generated by the mode-switching
dynamics \eqref{eq:policy}--\eqref{eq:hazard}. The objective integrates
compliance across both modes; because the action path affects both the
rate of compliance accumulation and the timing of the mode switch through
\eqref{eq:exposure}, \eqref{eq:optproblem} is not separable across time in
general. The proof-of-concept restricts attention to uniform investment
$a_t = \bar{C}/T$ and optimizes over $\sv$ only. Joint optimization over
$\{a_t\}$, in particular, front-loading versus spreading investment to
manage exposure, is a natural extension.

\section{Cost-Effectiveness Metrics}
\label{sec:assessment}

The optimization structure supports cost-effectiveness comparison across
mechanisms on a common footing. Three metrics follow from the model:
\begin{itemize}
  \item \textbf{Compliance efficiency:} cumulative compliance per unit of
        total investment, $\bigl(\sum_t y_t\bigr)/\bar{C}$. Because the
        proof-of-concept fixes total investment at $\bar{C}$ with uniform
        allocation, this is a constant rescaling of cumulative compliance;
        the two are not independent axes.
  \item \textbf{Time to threshold:} the number of periods to reach a
        specified compliance level $\bar{y}$, capturing how quickly a
        mechanism loads and, under the mode-switching dynamics, whether
        it reaches $\bar{y}$ before the contested mode engages. The 
        metric is right-censored at $T$: mechanisms that never reach 
        $\bar{y}$ within the horizon are recorded as censored rather than 
        assigned a numerical time, and are excluded from comparisons on 
        this metric.
  \item \textbf{Durability ratio:} compliance maintained after $A$
        withdraws investment relative to peak compliance, operationalized
        in the PoC by zeroing investment at $T/2$ and measuring residual
        compliance at $T$. It is governed jointly by $\rho(\sv)$, $k^*$,
        and whether the system has entered the contested mode.
\end{itemize}
The durability ratio is particularly important for policy evaluation.
Mechanisms with high upfront cost but high durability (e.g., standards
entrenchment) may dominate lower-cost mechanisms requiring continuous
investment over long horizons. A mechanism that accumulates durable 
compliance \emph{before} detection can retain it after the government 
mobilizes, whereas one that relies on ongoing permissive-mode 
transmission collapses once the sign reverses. For applied use, 
mechanisms can be displayed on a cost-effectiveness Pareto frontier over 
compliance efficiency and durability, allowing policymakers to compare 
instruments on cost-effectiveness and durability without engaging the 
underlying mathematics.

\section{Mechanism Illustrations}
\label{sec:examples}

We illustrate the framework with two stylized mechanisms that differ
sharply in their switch-vector configurations and consequent dynamics.
These are analytic illustrations, not empirical case studies, of how the
same formal structure generates qualitatively different predictions. In 
particular, the demonstrate how the mode-switching machinery distinguishes 
mechanisms that a single-mode model would rank identically.

Each illustration below highlights the switch dimensions most salient to 
the mechanism. Dimensions left unstated take non-distinctive or default 
values and are omitted for readability. The proof-of-concept implementation 
(Section~\ref{sec:poc}) operationalizes a six-dimension subset of the 
switch vector; accordingly, the switch tuples in 
Table~\ref{tab:mech_scores} report those six coordinates.

\subsection{Mechanism I: An Emigration Option Contract}
\label{sec:mech1}
\begin{example}[Emigration contract]
Country $A$ offers citizens of Country $B$ a guaranteed immigration
pathway valid for a defined window, at zero cost to exercise. No mass
emigration need occur for the mechanism to exert influence.
\end{example}

\noindent The switch vector is characterized as follows: \emph{influence
timing} latent-first (the unexercised option does the work);
\emph{target population} networked (the mobility-sensitive professional
stratum); \emph{propagation pathway}
cascade (elite behavior shapes institutional preferences broadly);
\emph{behavioral channel} divestment plus alignment-seeking;
\emph{momentum persistence} low (the option expires);
\emph{locus of pressure} civil society ($p_t^P$ primary); \emph{cost
bearer} $B$ (lost human capital); \emph{legibility} obscured (appears as
a standard immigration program). In the model this produces low
$\mu(\sv)$, moderate direct injection $\eta(\sv)$, and low $\rho(\sv)$,
with pressure loading primarily on $p_t^P$.

The mode-switching structure gives this mechanism a distinctive profile.
Because its cost falls on $B$, the government channel carries little of 
the mechanism's weight ($w^G$ small), so Proposition~\ref{prop:legibility}
implies it is relatively insulated from legibility penalties: the
contested-mode reversal acts on a channel that barely enters its 
influence. Its latent option also loads momentum steadily through $\eta$, 
so compliance accumulates during the permissive window and its structural 
stock retains a modest fraction of that compliance after detection.

\subsection{Mechanism II: A Conditional Liquidity Corridor}
\label{sec:mech2}
\begin{example}[Conditional liquidity corridor]
Country $A$ establishes credit facilities for $B$'s private banking
sector whose cost varies continuously with $B$'s regulatory alignment, as
measured by a published policy index.
\end{example}

\noindent The switch vector differs sharply: \emph{momentum persistence}
low (the mechanism requires ongoing investment; if $A$ withdraws the
facility it collapses); \emph{locus of pressure} shifts to $p_t^F$ ($B$'s
firms bear the cost gradient and lobby government directly);
\emph{behavioral channel} alignment-seeking; \emph{legibility} obscured
(the index is public but the causal chain from index to credit cost to
lobbying is not); \emph{cost bearer} shared ($A$ sustains the facility 
while $B$'s firms bear compliance costs). Structural stock $k_t$ 
accumulates as $B$'s firms build compliance infrastructure, but since it 
accrues to $B$'s private sector rather than to $A$, $A$ must sustain 
$a_t$ to maintain influence.

Both mechanisms are obscured, so by \eqref{eq:exposure} they share the 
same exposure rate $\phi(\ell)$; what differs is how fast each drives 
exposure. The emigration contract loads momentum through its latent 
option (higher $\eta$), so its exposure stock $E_t$ accrues faster and it 
reaches the contested mode \emph{sooner} (dwell $13$ versus $15$). During 
its shorter permissive window it accumulates more compliance, and its 
structural stock retains more of that compliance after detection than the 
liquidity corridor, whose low persistence ($\rho$) lets permissive-mode 
gains decay once the government channel reverses. The two-phase model 
separates the mechanisms in a way a single-mode model cannot: by showing 
that a faster-loading mechanism, despite being detected earlier, can 
remain more effective throughout. 

\begin{figure}[ht]
  \centering
  \includegraphics[width=0.82\textwidth]{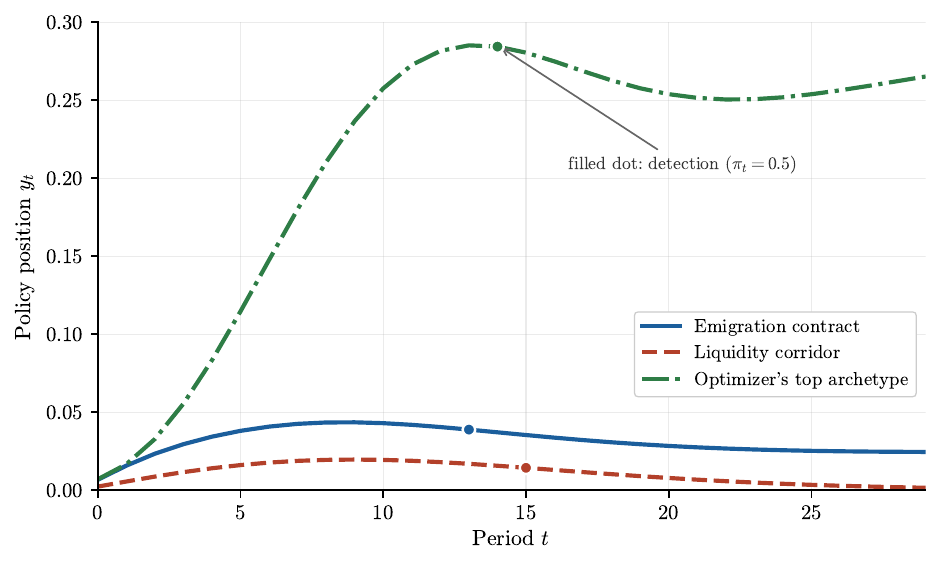}
\caption{Policy-position trajectories $y_t$ over the $T=30$ horizon for
    the two illustrative mechanisms of Section~\ref{sec:examples} and the
    optimizer's top-ranked archetype. Filled dots mark each mechanism's
    dwell time: the period at which the contested-mode probability $\pi_t$
    first crosses $0.5$. Because $\pi_t$ rises continuously from the first
    period, the government-channel reversal is blended in gradually via
    \eqref{eq:blend}, and each trajectory turns over once rising resistance
    $\beta^{\mathrm{eff}}_t\,r(y_t,k_t)$ overtakes the forcing term
    $\Phimech^{\mathrm{eff}}_t\,m_t\,\delta$. For the two low-persistence
    mechanisms this occurs several periods ahead of the dot; for the
    high-persistence archetype, whose momentum keeps the forcing term
    growing, it roughly coincides with it. The aftermaths then separate on
    persistence and cost incidence: the liquidity corridor decays to near
    zero; the emigration contract retains roughly half its peak, since its
    cost falls on $B$ and the reversing government channel therefore 
    carries little of its weight ($w^G$ small, cf.\
    Proposition~\ref{prop:legibility}); and the top archetype partially
    recovers as accumulated momentum reasserts pressure. The two
    illustrative mechanisms are hand-specified. Under the illustrative 
    lookup values of Section~\ref{sec:poc} the whole archetype space sits 
    below the compliance ceiling; the vertical scale is set accordingly.}
  \label{fig:trajectories}
\end{figure}

\begin{table}[ht]
  \centering
  \caption{Illustrative Mechanisms Located in the PoC Switch Space}
  \label{tab:mech_scores}
  \renewcommand{\arraystretch}{1.3}
  \begin{tabular}{L{2.7cm} L{4.4cm} L{1.6cm} L{1.9cm} L{1.2cm} L{1.1cm}}
    \toprule
    \textbf{Mechanism} & \textbf{Switch vector (6-dim.)}
      & \textbf{Cumul. compl.} & \textbf{Durability}
      & \textbf{Dwell} & \textbf{Rank} \\
    \midrule
    Emigration contract
      & Latent / Civil society / Cascade / Low / Obscured / B
      & 0.941 & 0.077 & 13 & 287 \\
    Liquidity corridor
      & Active / Private sector / Direct / Low / Obscured / Shared
      & 0.325 & 0.019 & 15 & 638 \\
    \bottomrule
  \end{tabular}
  \smallskip\\
  {\small\textit{Note.} The two mechanisms of Section~\ref{sec:examples}
    are hand-specified narrative archetypes, not optimizer selections.
    Their ranks (of $972$) reflect that they illustrate the framework's 
    expressiveness. Both carry \emph{Low} momentum persistence, so 
    both shed most compliance after investment withdrawal (low 
    durability), consistent with the Section~\ref{sec:examples} 
    narrative. \emph{Dwell} is permissive-mode dwell time in periods. 
    Values are illustrative, from the uncalibrated mode-switching 
    parameterization. Metrics are defined in
    Section~\ref{sec:assessment}. The switch vector is in the order 
    influence timing / locus of pressure / propagation pathway / momentum 
    persistence / legibility / cost bearer.}
\end{table}

\section{Proof of Concept}
\label{sec:poc}

To demonstrate the pipeline's tractability, we implemented a reduced-form
version of the framework as a local simulation tool. The proof-of-concept
(PoC) uses a six-switch vector (\emph{influence timing}, \emph{locus of 
pressure}, \emph{propagation pathway}, \emph{momentum persistence},
\emph{legibility to $B$}, \emph{cost bearer}) with three to four candidate
values per switch, yielding $3^5 \times 4 = 972$ enumerable
combinations.\footnote{Five ternary switches and one quaternary switch
(locus of pressure).} Section~\ref{sec:poc-sim} exercises the simulator 
and optimizer over the full switch space. Section~\ref{sec:llm-stage} tests 
the optional generation stage of Figure~\ref{fig:pipeline} over the same 
space, asking whether an archetype's design choices survive translation into 
a written instrument.

\subsection{Simulation Studies}
\label{sec:poc-sim}

\paragraph{Design.} The PoC evaluates the mode-switching model of
Section~\ref{sec:model} through three studies that share one simulator
and vary different inputs. Table~\ref{tab:studies} summarizes the design.

\begin{table}[ht]
  \centering
  \caption{Simulation Studies: Design}
  \label{tab:studies}
  \renewcommand{\arraystretch}{1.3}
  \begin{tabular}{L{2.6cm} L{4.3cm} L{4.6cm}}
    \toprule
    \textbf{Study} & \textbf{What varies} & \textbf{What is held fixed} \\
    \midrule
    Scenario sweep
      & Scenario parameters $\alpha,\beta,\delta,\kappa$, perturbed one
        at a time about baseline
      & Lookup tables; initial conditions at baseline \\
    Lookup robustness
      & Switch-to-parameter lookup values, perturbed within ordered tiers
        ($100$ draws)
      & Scenario parameters at selected configurations \\
    Initial-condition map
      & Initial policy $y_0$ and entrenchment $k_0$ over a $7\times7$ grid
      & Scenario parameters and lookup tables at baseline \\
    \bottomrule
  \end{tabular}
  \smallskip\\
  {\small\textit{Note.} All three studies re-run the full optimizer
    ($972$ combinations) at each grid point or draw. The scenario sweep
    is reported here in one-at-a-time form; a full $625$-cell factorial
    grid over the four scenario parameters is implemented and available
    in the accompanying code. Total wall-clock time for the three studies
    is approximately $10$ seconds using an AMD Ryzen AI 9 HX 370 machine 
    running Windows 11.
    }
\end{table}

\paragraph{Illustrative functional forms.} The simulation adopts
$r(y_t,k_t)=y_t/(1+\alpha k_t)$, which satisfies the sign restrictions of
Section~\ref{sec:dynamics} and admits a clean within-mode closed form. It
is chosen for tractability, not empirical grounding; the true form is an
open question.

\paragraph{Notional parameter mappings.} The switch-to-parameter lookup
tables assign numerical values to $\rho,\mu,\eta,\nu,\wv$, and $\Lam_0$ 
from qualitative reasoning about the switch dimensions, together with the
mode-switching quantities: the legibility modifier $g(\cdot)$ and
exposure rate $\phi(\cdot)$, the contested-mode reversal $\gamma$, and the
resistance ratio $\beta(\mathrm{C})/\beta(\mathrm{P})$. These values 
demonstrate that the pipeline distinguishes mechanisms and produces a 
sensible ranking.

\begin{table}[ht]
  \centering
  \caption{Baseline parameter values for the proof-of-concept. These are
    illustrative choices to demonstrate tractability, not calibrated
    estimates (Remark~\ref{rem:functional_forms}).}
  \label{tab:baseline_params}
  \renewcommand{\arraystretch}{1.2}
  \begin{tabular}{L{5.5cm} l l}
    \toprule
    \textbf{Parameter} & \textbf{Symbol} & \textbf{Value} \\
    \midrule
    Planning horizon            & $T$        & $30$ \\
    Budget                      & $\bar{C}$  & $1.0$ \\
    Resistance curvature        & $\alpha$   & $1.0$ \\
    Baseline resistance (permissive) & $\beta(\mathrm{P})$ & $0.5$ \\
    Contested resistance ratio  & $\beta(\mathrm{C})/\beta(\mathrm{P})$ & $2.0$ \\
    Policy inertia              & $\chi$     & $0.85$ \\
    Credibility                 & $\delta$   & $0.8$ \\
    Entrenchment decay          & $\kappa$   & $0.1$ \\
    Detection scale             & $\tau$     & $5.0$ \\
    Investment visibility       & $\omega$   & $1.0$ \\
    Threshold (time-to-threshold) & $\bar{y}$ & $0.5$ \\
    \bottomrule
  \end{tabular}
\end{table}

\paragraph{Results.}

The PoC computes, for each of the $972$ archetypes, equilibrium and
cumulative compliance, time to threshold, durability ratio, and
permissive-mode dwell time, and ranks archetypes accordingly. It 
establishes that the inverted regime-switching model is computable end to 
end, not that any particular mechanisms are preferable. 

The mechanism space is \emph{differentiated}: archetypes vary
materially in their metrics, so the scoring resolves distinctions among 
designs. The two quantities that exist only because the model switches 
modes both take a range of values across the space. The durability ratio 
takes $577$ distinct values spanning $[0.000, 0.477]$, and permissive-mode 
dwell time takes $24$ distinct values spanning $7$ to $30$ periods, where 
the upper value coincides with the horizon $T=30$ and denotes mechanisms 
not yet detected within it (right-censored, as with time to threshold). 
Moreover, a Pareto frontier over compliance efficiency and durability
retains $9$ of the $972$ archetypes, indicating that the space contains 
genuine trade-offs.

The pipeline responds to its inputs in a structured way. Varying
scenario parameters, lookup values, and initial conditions produces
systematic variation in the ranked output, which is what the three sweeps 
of Table~\ref{tab:studies} demonstrate. Note that their specific rankings 
are artifacts of the illustrative parameterization.

The PoC does \emph{not} establish that any particular mechanism archetype
is optimal, nor that the mechanism space partitions into interpretable
regimes. Whether such structure exists is a question for a calibrated
model, not this demonstration. Regularities that do appear in the sweeps
are analytic consequences of
Propositions~\ref{prop:credibility}--\ref{prop:durability}, not 
empirical discoveries. We treat them as confirmation that the
implementation is faithful to the model. The clearest instance is that
the top-ranked cluster is constant on two dimensions, taking
\emph{momentum persistence} $=$ High and \emph{cost bearer} $=$ $B$
throughout. This is exactly what
Propositions~\ref{prop:legibility}--\ref{prop:durability} predict: high
persistence retains permissive-mode gains, and cost incidence on $B$
lowers $w^G$ and so insulates against the legibility penalty.

\paragraph{Implementation notes.} Under the illustrative lookup values 
the entire archetype space sits well below the compliance ceiling. The 
highest peak policy position of any of the $972$ archetypes is 
approximately $0.29$, so the $y^*=1$ clip in \eqref{eq:ystar} never binds 
and every mechanism remains in a low-compliance regime. This reflects the 
provisional parameterization, not a property of the framework. It does 
not affect the proof-of-concept's claims, which are \emph{ordinal} 
(differentiation, ranking, and the cost-effectiveness frontier). 

Differentiation among archetypes rests on transient speed, dwell time, 
and durability. The lookup-robustness sweep perturbs values \emph{within} 
ordered tiers, so it exercises the pipeline against magnitude variation 
but holds the ordinal structure of the qualitative assignments fixed. The 
ordinal assumptions themselves are a calibration question. 

The policy update \eqref{eq:policy} carries the inertia term $\pinertia$, 
whose interaction with baseline resistance $\beta$ must satisfy a mild 
stability condition for initial conditions to influence trajectories. 
For example, writing the policy self-map at low entrenchment ($r\approx y$) 
as $y_{t+1}\approx(\pinertia-\beta)\,y_t+\text{(forcing)}$ means that 
genuine and stable dependence on initial conditions requires
$0<\pinertia-\beta<1$. Below the lower bound the update suppresses the
initial position; above the upper bound it fails to contract. 

\subsection{Generation-Stage Recoverability}
\label{sec:llm-stage}
The optional generation stage of Figure~\ref{fig:pipeline} translates an
archetype into a concrete candidate instrument. Obviously, translation is 
only useful if the instrument still carries the archetype's design choices. 
We test this by asking whether the six switch values can be recovered from
the generated instrument alone. A dimension that is recovered reliably is
one the written instrument encodes; a dimension that is not is one the
archetype leaves underdetermined once it is turned into prose.

\paragraph{Design.} The study covers all $972$ archetypes of the
simulation's switch space and is fully crossed over four language models
from four developers, three hosted and one open-weight model served
locally (Table~\ref{tab:recov_matrix}). Each model acts as both generator
and classifier. For each archetype, each generator wrote three candidate
instruments in independent end-to-end replicates, yielding $11{,}664$
instruments, and each instrument was classified once by each of the four
models, yielding $46{,}656$ classifications. A fixed generation prompt
(Appendix~\ref{app:prompt.gen}) supplied the six switch values in the
synthetic Civilization~A/B setting, required a structured description
(name, narrative, activation trigger, expected behavioral response, and
durability profile), and forbade switch vocabulary in those fields. The
classification prompt (Appendix~\ref{app:prompt.blind}) received only
those fields, never the source archetype or the generator's report of its
own choices, and assigned one value per dimension using the definitions
of Table~\ref{tab:switches}. Every classifier received an identical prompt
for a given instrument. All models ran at provider-default sampling with
reasoning effort set to medium; the three replicates capture run-to-run
variation rather than suppressing it.

\paragraph{Scoring.} 
A dimension is recovered when the classifier's value matches the
archetype's exactly, so each classification yields six judgments. We
report \emph{recovery} as the percentage of these judgments that are
correct, and \emph{exact-vector recovery} as the percentage of
classifications in which all six are correct. Pairs that use the same 
model to generate and classify share priors and could agree for reasons 
unrelated to the text, so we treat the twelve cross-model pairs as the 
primary result and report same-model pairs separately. API errors and
unparseable responses ($142$ of $46{,}656$, or $0.3\%$) were retained and
scored as recovering nothing. Because the $36$ cross-model
classifications of an archetype share one source specification,
confidence intervals resample archetypes ($2{,}000$ bootstrap draws over
the $972$ vectors). Uniform guessing would recover $31.9\%$ of dimensions
and $0.1\%$ of complete vectors.

\paragraph{Results.} 
Across the $34{,}992$ cross-model classifications, classifiers recovered
$76.0\%$ of design choices (95\% CI $75.5$--$76.6\%$), on average $4.6$
of the six per classification, against $31.9\%$ by chance. They recovered
all six together in $19.9\%$ of classifications ($18.9$--$21.1\%$),
against $0.1\%$ by chance. When the same model generated and classified,
recovery rose only to $77.6\%$, so shared priors account for little of
the agreement. The result was stable across the three replicates
($75.7$--$76.4\%$) and does not rest on any single model: the twelve
cross-model pairs range from $68.6\%$ to $84.6\%$
(Table~\ref{tab:recov_matrix}). Instruments written by the open-weight
model were the least recoverable ($71.7\%$, against $77.3$--$77.8\%$ for
the hosted generators). The classifiers fell into two tiers: GPT and
Gemini recovered $78.4$--$78.8\%$ of design choices, Qwen and Claude
$73.2$--$73.8\%$. The Claude classifier's figure includes $133$ failed
calls, $124$ of them on instruments written by Gemini; excluding them
raises it to $74.3\%$ without changing the ordering.

\begin{table}[ht]
  \centering
  \caption{Recoverability: Percentage of Dimension Recovery by Generator 
  and Classifier}
  \label{tab:recov_matrix}
  \renewcommand{\arraystretch}{1.2}
  \begin{tabular}{l cccc c}
    \toprule
    & \multicolumn{4}{c}{\textbf{Classifier}} & \\
    \cmidrule(lr){2-5}
    \textbf{Generator} & Claude & Gemini & Qwen & GPT
      & \textbf{Cross-model mean} \\
    \midrule
    Claude & \textit{76.6} & 77.9 & 74.2 & 80.1 & 77.4 \\
    Gemini & 75.4 & \textit{83.2} & 75.5 & 82.5 & 77.8 \\
    Qwen   & 68.6 & 72.6 & \textit{69.3} & 73.9 & 71.7 \\
    GPT    & 75.5 & 84.6 & 71.6 & \textit{81.3} & 77.3 \\
    \midrule
    \textbf{Cross-model mean} & 73.2 & 78.4 & 73.8 & 78.8
      & \textbf{76.0} \\
    \bottomrule
  \end{tabular}
  \smallskip\\
  {\small\textit{Note.} Each cell is the mean percentage of the six
    dimensions recovered over $2{,}916$ classifications ($972$ archetypes
    $\times$ $3$ replicates). Italic diagonal cells use the same model to
    generate and classify and are excluded from the cross-model means.
    Cross-model exact-vector recovery is $19.9\%$. Models, as called:
    \texttt{claude-sonnet-5} (Anthropic), \texttt{gemini-3.8-flash}
    (Google), \texttt{gpt-5.6-sol} (OpenAI), and \texttt{qwen3.8-27b}
    (open weights, served locally).}
\end{table}

\paragraph{Where specification is lost.} Recovery varies sharply by
dimension (Table~\ref{tab:recov_dims}). Momentum persistence is
recovered most reliably and propagation pathway least. The errors are
directional: classifiers over-assign the broad or intermediate value of a
dimension. Cascade accounts for $57\%$ of propagation calls against a
true share of $33\%$, a mixed locus for $41\%$ of locus calls against
$25\%$, and medium persistence for $44\%$ of persistence calls against
$33\%$. Specification is lost at the narrow values (direct propagation, 
deniable legibility, a civil-society locus, latent timing, low persistence) 
and absorbed by their broader neighbors, and the high recovery of the broad 
values partly reflects this tendency. Chance-corrected agreement ($\kappa$ 
in Table~\ref{tab:recov_dims}) adjusts for it and keeps propagation and 
legibility as the weakest dimensions. Three of the losses have distinct 
sources.

\emph{Direct propagation is lost in every pairing.} It was recovered in
$34\%$ of cases, and in no cross-model cell above $49\%$; the remainder
was read as cascade ($41\%$) or systemic ($25\%$). Instruments in which
pressure lands on its targets without spreading among them are
evidently hard to write, hard to recognize, or both, and the archetype
does not reliably survive either way.

\emph{Deniability depends on the writer.} Transparent legibility was
recovered in at least $99\%$ of cases in every cell, deniable legibility
in $53\%$, with a third of deniable instruments read as obscured.
Deniable instruments written by Claude and GPT were recovered in
$53$--$84\%$ of cases, those written by Gemini and Qwen in
$29$--$42\%$. Whether an instrument conveys that $B$ cannot attribute it
depends on how it is written.

\emph{The civil-society locus depends on the reader.} Civil-society
instruments were recovered in $59\%$ of cases, and most errors assigned
them a mixed locus; the Qwen classifier did so for $54\%$ of them. A
mechanism that originates pressure among $B$'s population but ultimately
presses on its government can defensibly be read either way. This is the
boundary between the stratum a mechanism addresses and the subsystem
through which it routes pressure (Section~\ref{sec:switch}), which the
six-dimension PoC represents only through the locus of pressure.

The cost bearer is recovered in $81\%$ of classifications. Its main error, 
reading costs borne by $B$ as shared ($20\%$), is consistent with an 
ambiguity between who funds a mechanism and who bears its consequences 
that the dimension's definition does not resolve.

\begin{table}[ht]
  \centering
  \caption{Recoverability: Percentage of Recovery by Dimension and True 
    Value (Cross-Model Pairs)}
  \label{tab:recov_dims}
  \renewcommand{\arraystretch}{1.2}
  \begin{tabular}{L{3.9cm} c c L{7.3cm}}
    \toprule
    \textbf{Dimension} & \textbf{Recovery $\%$} & $\boldsymbol{\kappa}$
      & \textbf{Recovery by true value $\%$} \\
    \midrule
    Influence timing     & 78.6 & 0.68
      & Latent 71.7; Active 75.8; Mixed 88.2 \\
    Locus of pressure    & 77.9 & 0.71
      & Civil society 58.6; Government 74.1; Private sector 86.5;
        Mixed 92.5 \\
    Propagation pathway  & 59.4 & 0.40
      & Direct 33.8; Systemic 56.8; Cascade 87.8 \\
    Momentum persistence & 87.1 & 0.81
      & Low 73.1; High 90.5; Medium 97.9 \\
    Legibility to $B$    & 72.5 & 0.59
      & Deniable 52.6; Obscured 65.0; Transparent 99.8 \\
    Cost bearer          & 80.7 & 0.72
      & $B$ 75.9; Shared 83.0; $A$ 83.3 \\
    \bottomrule
  \end{tabular}
  \smallskip\\
  {\small\textit{Note.} Recovery is the fraction of the $34{,}992$
    cross-model classifications that assign the archetype's value, with
    failed calls scored as unrecovered. Each true value of a ternary
    dimension appears in $11{,}664$ classifications, each locus value in
    $8{,}748$. $\kappa$ is Cohen's kappa on valid responses ($0$ is 
    chance-level agreement, $1$ is perfect agreement), which
    corrects for how often classifiers assign each value. Values within
    a row are ordered from least to most recoverable.}
\end{table}

\paragraph{Interpretation.} The generation stage runs at scale and can
be audited mechanically. The archetypes are substantially but not fully
recoverable from what it produces. In four of every five cross-model
classifications at least one design choice was not recovered. A 
generated instrument should not be assumed to instantiate the archetype
it was asked for; it should be classified back before it is used. The 
dimension-level losses suggest what makes a switch usable by the
pipeline's models. Values defined by an observable feature of an
instrument (a transparent mechanism, spread among actors) were
recovered; values defined by an absence or by a judgment of primacy
(direct propagation, deniability, or a civil-society) were absorbed by 
their neighbors. 

The switches here are notional, so these values are illustrations. The 
general requirement is that a switch vector handed to a generation stage 
needs definitions stated as observable markers, with decision rules between 
adjacent values, and recoverability is the test of whether it has them. The 
classifiers are themselves language models, so recovery measures agreement 
between model writers and model readers. 

Human coding of a stratified subset is a natural next step. The crossed 
design separates writer from reader effects only in aggregate. And the 
models are version-contingent, so the committed records, not the procedure, 
are the reproducible artifact. Nothing here establishes that a generated 
instrument would work against a real target (Section~\ref{sec:gap}). Human 
coding would support auditing the prompting stage and inform future 
engineering efforts. 

\subsection{Transition to Empirical Calibration}
\label{sec:calibration}
The PoC demonstrates that the pipeline architecture is computationally
tractable, produces discriminating output under illustrative
parameterization, and has a generation stage that can be audited at
scale. Transition to a calibrated model requires (i) empirical
estimation or expert elicitation of the switch-to-parameter mappings,
\emph{including} the new mode-switching quantities ($g,\phi,\gamma,\tau$,
and the resistance ratio); (ii) selection of functional forms for $r$,
the mode-transition hazard, and the mapping functions, guided by case
evidence; (iii) validation of predictions, especially dwell-time and
durability predictions, against historical episodes of geoeconomic
influence; and (iv) restating the switch definitions in 
machine-readable form, with observable markers and decision rules for 
each value, so that the generation stage applies them consistently 
(Section~\ref{sec:llm-stage}). A calibrated model will support more 
robust development of the generation stage components. 

\section{Discussion}
\label{sec:discussion}

The framework's contribution is to account for the internal transmission 
structure through which indirect influence operates. It models the moment 
that structure changes as when the target detects the mechanism and its 
government channel reverses. We discuss scope, generalizability, the 
analytic-versus-emergent character of the results, and the framework's 
ethical posture.

\subsection{Scope Conditions}
\label{sec:scope}
We have positioned the IGI approach as a means of geopolitical influence.
The same architecture may apply to the discovery of policies for internal
deployment by a nation, a firm, or another large-scale organization, and
possibly as a tool for designing policies in computer-science settings.
As specified, however, it is best suited to target nations with
meaningful internal political economies. Targets without sufficiently
distinct and autonomous civil societies, private sectors, and government
factions cannot route the pressure a mechanism creates.

\citet{lektzian2007institutional} show that sanctions are systematically
less effective against nondemocratic targets, arguing that the key to
success is generating political costs for the regime's winning coalition.
In autocracies that coalition is small, rents accrue to insiders, and the
internal transmission channels are weak or absent. The IGI framework 
formalizes this: in the limiting case where $B$'s civil society and 
private sector exert no independent influence, the pressure vector 
collapses toward the government subsystem, and, because that channel is 
precisely the one that reverses in the contested mode, external 
mechanisms lose their primary lever the moment they are detected. 

The unitary-actor assumption of the formal sanctions literature 
\citep{Drezner1999} is a reasonable approximation for a specific and 
important class of policy targets. The framework's contribution is to make 
explicit when that approximation holds and when it breaks down. 
\citet{rochat2026effectiveness} document that targeted sanctions 
contributed to reform in Myanmar, where internal channels were sufficiently 
autonomous to transmit pressure, but failed in Zimbabwe, where the regime's 
entrenchment blocked those channels. Such dynamics are what the IGI can 
predict via its generative capacity. 

\subsection{Generalizability}
\label{sec:general}
The underlying logic (engineering an incentive environment so that a
target system's own actors generate the desired pressure) extends beyond
the interstate case. 

\emph{Domestic supply-chain resilience}: inducing 
private-sector firms to internalize national-security considerations in 
sourcing without direct mandates is structurally identical to the 
foreign-policy case. The locus of pressure shifts to the private-sector 
subsystem, the behavioral channel becomes alignment-seeking, and the 
designer is a domestic government. The switch-vector approach accommodates 
this without structural modification, and the mode-switching structure 
applies too: firms and regulators ``detecting'' and resisting an incentive 
scheme is the same detection-triggered transition. 

\emph{Conditional development finance}: IMF structural conditionality 
instantiates the liquidity-corridor mechanism. It is an external actor 
offering financial access whose terms vary with policy compliance, 
routing pressure through domestic financial institutions that lobby the 
government. The conditionality literature documents precisely the 
internal transmission dynamics the model predicts 
\citep{mayer2005political,dreher2009imf}.

In a historical connection, \citet{schoppa1997bargaining} documents
that synergistic strategies (deliberately activating aligned domestic
coalitions) deployed in U.S.--Japan trade negotiations had measurable
effect. Their success tracked the availability of internal coalitions, 
consistent with the IGI framework's pressure-routing predictions.

\subsection{Analytic versus Emergent Results}
\label{sec:analytic-emergent}
Some of the framework's regularities are \emph{analytic}: they follow 
directly from the equilibrium conditions and comparative statics of 
Section~\ref{sec:model}. That compliance rises with credibility
(Proposition~\ref{prop:credibility}) and that deniable, high-persistence
mechanisms score well in the permissive mode are corollaries of the
model's monotone structure, not discoveries of the simulation. A
simulation that ``found'' them would merely be confirming the algebra. We 
treat their appearance in the sweeps as a check that the implementation
is faithful. Other questions are \emph{empirical} and cannot be read off
the equilibrium conditions, including which archetype is preferable as a
function of initial conditions, and whether the permissive/contested 
trade-off makes selection depend on a scenario rather than being dominated 
by a single archetype. Resolving these requires a calibrated model, and 
under the illustrative parameterization no such selection structure should 
be inferred. 

\subsection{Sensitivity to Parameterization}
\label{sec:sensitivity}
The ranking of archetypes depends on the switch-to-parameter mappings,
which in the PoC are illustrative. Varying parameter values around 
baseline identifies which parameters most influence the ranking, and 
which mappings most urgently require calibration. This is analogous to 
prior-sensitivity analysis in Bayesian models. The framework is designed 
to make this diagnostic transparent: as case evidence, natural experiments, 
or expert elicitation constrain the parameter space, the archetype rankings 
update without altering the model's structure. Under the mode-switching 
model the same logic extends to the transition quantities. The detection 
threshold $\tau$ and exposure rate $\phi$ are prime elicitation targets 
precisely because dwell time is where the interesting trade-offs live.

\subsection{The Archetype-to-Instrument Gap}
\label{sec:gap}
The pipeline produces abstract archetypes that are structured descriptions 
of how influence is timed, routed, and sustained. Translating an archetype 
into an actionable instrument requires contextual judgment about the target's
institutions, culture, and political economy. The optional language-model 
generation stage assists this translation, but introduces a validity 
question: does the generated description faithfully reflect the archetype's 
structural properties, including its mode-switching profile, or does it 
drift? Section~\ref{sec:llm-stage} measures one face of this gap, whether an 
archetype's design choices can be recovered from the instrument written for 
it, and finds them substantially but not fully recoverable. It does not 
measure whether a faithfully written instrument would function as its 
archetype predicts.

Conflict simulation and structured expert review are the natural 
validation mechanisms. Participants can probe whether a candidate 
instrument would route pressure through the specified channels and 
whether its credibility, persistence, and, crucially, detectability 
properties survive contact with real institutional constraints. 

\subsection{Ethical Concerns}
\label{sec:ethics}
The IGI framework is analytic, not prescriptive. It characterizes how
indirect mechanisms produce compliance; it does not recommend their use,
and it does not adjudicate the legality or legitimacy of any mechanism the
optimizer surfaces. Its mechanisms may be lawful, unlawful, or of 
ambiguous status under domestic and international law. The model's 
structure rewards mechanisms that are deniable and that shift costs onto 
the target's own population and firms. Those are designs whose ethical
and legal character demands scrutiny that a compliance-maximizing
objective does not supply. 

The same machinery that identifies effective influence mechanisms 
identifies the mechanisms a \emph{defender} should anticipate and harden 
against. The framework's defensive use (mapping one's own vulnerabilities 
to indirect influence) is at least as natural as its offensive use. 
Whether used in an offensive or defensive settings, the optimizer's 
output should be treated as a hypothesis space to be filtered, not a set 
of recommendations to be executed. Making the deniability and 
cost-incidence properties explicit allows that filtering to occur.

\subsection{Limitations and Future Work}
\label{sec:future}
The switch vector is discrete by design, decomposing mechanism space into
a finite set of enumerable configurations. This reflects the
interpretability requirement (Section~\ref{sec:intro}): a discrete switch
vector lets analysts reason explicitly about configurations and compare
them on a common footing. However, many real mechanisms do not fall 
cleanly into a single switch value. A natural generalization replaces the
discrete switch vector with a continuous latent space, drawing on the
variational switching state-space framework of 
\citet{GhahramaniHinton2000}.

A nearer-term extension is a general protocol for writing
machine-readable specifications of a switch vector, independent of any
particular vector or application. The recoverability study of
Section~\ref{sec:llm-stage} supplies its core loop. First, each switch
value is defined by features observable in an instrument description,
with decision rules that separate adjacent values. Second, generative
models instantiate instruments across the switch space. Third, models
other than the generator classify those instruments blind, in a fully
crossed design that separates, in aggregate, failures that follow the
writer, failures that follow the reader, and failures common to every
pairing, which implicate the definition itself. Fourth, values with low
chance-corrected recovery, or with errors that flow consistently toward
a neighboring value, are redefined, and the loop repeats. Fixing the set 
of values based on an analyst's design needs, so recovery cannot be 
raised by merging disctinctions, and human-coding of a stratified subset 
to ground the model agreement with expert judgment would help ensure 
reliability of the loop. A specification that passes this loop can be 
handed to any generation stage with a measured expectation of fidelity.

The expected-trajectory hazard (Remark~\ref{rem:hazard}) is worth
extending to the full stochastic form with Monte-Carlo over switching
times. Doing so would reintroduce the Markov-switching character of 
\citet{Hamilton1989} in a designed rather than inferred form, let 
latent mechanism configurations recovered from historical episodes seed 
the generative optimizer, and unite the inferential and generative 
directions. 

The investment path $\{a_t\}$ could be optimized jointly with $\sv$ 
(Section~\ref{sec:opt_problem}). Because investment intensity drives 
exposure, the optimal path trades compliance accumulation against dwell 
time, a genuinely dynamic problem the uniform-investment PoC does not 
address.

The model is bilateral. $A$ acts alone and $B$ has no outside options, 
yet the empirical sanctions literature finds that third parties 
materially condition outcomes. Allied and adversary states can absorb or 
offset a sender's pressure \citep{Early2012,Early2015}, and sender 
coalitions perform differently from unilateral senders 
\citep{BapatMorgan2009}. In the present framework a third party would 
enter as a competing momentum source with its own routing vector, or as a 
reduction in the effective credibility $\delta$; both are natural 
extensions of the pressure vector.

\section{Conclusion}
\label{sec:conclusion}

We have presented a formal, extensible framework for indirect geoeconomic
influence that accounts for the target's internal transmission structure.
IGI's central commitment is that a target's political economy under 
influence evolves under mode-dependent rules (permissive before detection, 
contested after) and that the sender's mechanism design shapes the transition
between them. IGI inverts the standard regime-switching problem: the
transition kernel is designed rather than estimated, and the policy
problem becomes steering regime occupancy over a horizon. A structured
switch vector decomposes mechanism space, a combinatorial optimizer
searches it, and cost-effectiveness metrics place mechanisms on a common
frontier. A proof of concept demonstrates its tractability. A fully crossed 
study of the optional generation stage shows that language models can 
translate archetypes into instrument descriptions. Other models recover 
most, but not all, of the specified design choices. Better recovery is 
likely obtainable through optimization of the model prompting protocol. 
Because the pipeline hands its switch vector to models, each switch must be 
defined in a form a model can apply in writing an instrument from a value 
\emph{and} recognizing the value in an instrument. As a result 
recoverability is a direct test of whether a switch vector meets that 
standard.

\section*{Data and Code Availability}
\label{sec:availability}
The proof-of-concept implementation, the prompts, the $11{,}664$
generated instrument descriptions, the $46{,}656$ classification records
of the recoverability study (with a manifest giving source runs and
SHA-256 hashes), and the analysis scripts are available from the
corresponding author.

\bibliographystyle{plainnat}
\bibliography{geo_framework}

\appendix
 
\section{Generation and Classification Prompts}
\label{app:prompt}
 
This appendix records the prompts used in the recoverability study
(Section~\ref{sec:llm-stage}). All four models received the same two
prompts. For generation, the six bracketed values were filled with one
archetype from the $972$-vector space; for classification, the narrative
fields of one generated instrument were appended. Models ran at
provider-default sampling with reasoning effort set to medium.
Classification calls were made between 23 and 30 September 2026.
 
\paragraph{A note on vocabulary.} The prompts below are transcribed as
executed, and they and the study records name the routing dimension
\emph{target subsystem}, with values Government / Private sector / Civil
society / Mixed. It is the \emph{locus of pressure} dimension of
Table~\ref{tab:switches}, distinct from \emph{target population}; the
label is retained here so that the prompts are reproduced exactly.
 
\subsection{Generation Prompt}
\label{app:prompt.gen}
 
\begin{small}
\begin{verbatim}
You are assisting with an analytic exercise in a fully fictional
setting. There are two fictional polities: Civilization A (the
sender) and Civilization B (the target). A seeks to shift B's
dominant resource-allocation policy toward cooperative exchange,
but cannot mandate this directly; it must design an indirect
mechanism that routes pressure through B's internal subsystems so
that B's own actors generate the compliance pressure.
 
B has three internal subsystems: its Government (sets official
policy, can detect and resist), its Private sector (controls
resource infrastructure), and its Civil society
(mobility-sensitive population).
 
An abstract mechanism is specified by six structural design
choices (a "switch vector"):
 
- Influence timing: {Latent | Active | Mixed} - whether influence
  operates before explicit deployment (an unexercised option doing
  work) or only after.
- Target subsystem: {Government | Private sector | Civil society |
  Mixed} - which subsystem bears primary pressure.
- Propagation pathway: {Direct | Cascade | Systemic} - how
  pressure moves through B's internal system.
- Momentum persistence: {Low | Medium | High} - whether the
  mechanism self-sustains or requires ongoing investment.
- Legibility to B: {Transparent | Obscured | Deniable} - how
  easily B's government identifies and attributes the mechanism.
- Cost bearer: {A | B | Shared} - who absorbs the mechanism's
  primary cost.
 
Here is the archetype to instantiate:
 
    Influence timing:     <FILL>
    Target subsystem:     <FILL>
    Propagation pathway:  <FILL>
    Momentum persistence: <FILL>
    Legibility to B:      <FILL>
    Cost bearer:          <FILL>
 
Produce a concrete candidate instrument that faithfully
instantiates THIS archetype in the Civilization A/B setting. Do
not change the design choices; your instrument must reflect them.
 
Do not use the switch-dimension names or their values (e.g.
"momentum persistence", "latent", "deniable", "cascade") anywhere
in the narrative fields; describe the instrument concretely
instead. The taxonomy vocabulary belongs only in the
self_reported_switches field.
 
Return your answer as JSON only, with no prose outside the JSON,
in exactly this structure:
 
{
  "mechanism_name": "<short evocative label>",
  "narrative": "<3-5 sentence description of how the instrument
    works: who is targeted, what A does, how pressure routes to a
    policy outcome>",
  "activation_trigger": "<what sets it in motion>",
  "expected_behavioral_response": "<what B's actors do>",
  "durability_profile": "<what happens if A stops investing>",
  "self_reported_switches": {
    "influence_timing": "<the value you instantiated>",
    "target_subsystem": "<...>",
    "propagation_pathway": "<...>",
    "momentum_persistence": "<...>",
    "legibility_to_b": "<...>",
    "cost_bearer": "<...>"
  }
}
 
The "self_reported_switches" field must state which switch values
your instrument actually reflects, judged on its own merits - not
simply copied from the input. If your instrument departs from the
requested archetype on any dimension, report the value it actually
reflects.
\end{verbatim}
\end{small}
 
\subsection{Classification Prompt}
\label{app:prompt.blind}
 
For classification, each generated instrument was reduced to its
narrative fields (mechanism name, narrative, activation trigger, expected
behavioral response, and durability profile), with the source archetype
and the generator's self-reported switches removed. Every classifier
received an identical prompt for a given instrument. The rubric states
the definitions of Table~\ref{tab:switches} and adds no operational
criteria beyond them.
 
\begin{small}
\begin{verbatim}
You are classifying a described policy instrument against a fixed
taxonomy. You will see only a description of a candidate mechanism
in a fictional setting (Civilization A influencing Civilization
B). Judge only from the description.
 
Assign one value on each of six dimensions:
 
- influence_timing: Latent | Active | Mixed
  Latent = an offer/capability does its work while unexercised or
  dormant. Active = pressure only builds through current, live
  deployment. Mixed = both.
- target_subsystem: Government | Private sector | Civil society |
  Mixed
  Which of B's subsystems bears the primary pressure. Government =
  the state/executive. Private sector = firms/resource
  infrastructure. Civil society = the
  population/mobility-sensitive individuals. Mixed = no single
  subsystem dominates.
- propagation_pathway: Direct | Cascade | Systemic
  Direct = pressure lands on the target actors and stops there.
  Cascade = it spreads actor-to-actor through
  peer/social/competitive imitation. Systemic = it becomes
  embedded in shared infrastructure/standards affecting the whole
  system at once.
- momentum_persistence: Low | Medium | High
  Judge from what happens if A stops investing: Low = collapses
  quickly. Medium = persists for a while then fades. High =
  self-sustaining.
- legibility_to_b: Transparent | Obscured | Deniable
  How easily B's government can identify and attribute the
  mechanism to A. Transparent = B's government can plainly
  identify A's hand. Obscured = attribution is possible but
  requires effort. Deniable = B's government cannot attribute the
  mechanism to A.
- cost_bearer: A | B | Shared
  Who absorbs the mechanism's primary cost. A = the sender. B =
  the target (its actors/economy). Shared = both.
 
Return JSON only, no prose:
 
{
  "switch_vector": {
    "influence_timing": "<value>",
    "target_subsystem": "<value>",
    "propagation_pathway": "<value>",
    "momentum_persistence": "<value>",
    "legibility_to_b": "<value>",
    "cost_bearer": "<value>"
  },
  "rationale": "<one brief sentence per dimension>"
}
 
Here is the instrument to classify:
 
<narrative fields of one generated instrument>
\end{verbatim}
\end{small}
 
\paragraph{Scoring.} Recovered switch vectors were compared to the true
archetypes dimension by dimension, by exact match on the six categorical
values with no partial credit; the rationale and self-report fields were
not scored. API errors ($128$) and unparseable responses ($14$) were
retained and scored as recovering none of the six dimensions.
Section~\ref{sec:llm-stage} reports the results.
 
\end{document}